\documentclass{article}

\usepackage[left=2cm, right= 2cm, top= 2cm, bottom= 2cm]{geometry}
\usepackage{pdflscape}
\usepackage{palatino}

\usepackage{graphicx} 
\usepackage{amsmath}
\usepackage{mathtools}
\usepackage{amsfonts}
\usepackage{amssymb}
\usepackage{amsthm}
\usepackage{stackrel}
\usepackage{float}

\usepackage{hyperref}

\usepackage[title]{appendix}

\usepackage[dvipsnames]{xcolor}

\usepackage[nameinlink]{cleveref}
\Crefname{section}{Section}{Sections}
\Crefname{appendix}{Appendix}{Appendices}
\Crefname{Algo}{Algorithm}{Algorithms}
\Crefname{Def}{Definition}{Definitions}
\Crefname{Thm}{Theorem}{Theorems}
\Crefname{Pro}{Problem}{Problems}
\Crefname{Prop}{Proposition}{Propositions}
\Crefname{Lem}{Lemma}{Lemmas}
\Crefname{Rm}{Remark}{Remarks}

\newtheorem{Def}{Definition}
\newtheorem{Thm}{Theorem}

\newtheorem{Prop}{Proposition}

\newtheorem{Coro}{Corollary}
\newtheorem{Lem}{Lemma}
\newtheorem{Rm}{Remark}

\newcommand{\ket}[1]{|{#1}\rangle}
\newcommand{\bra}[1]{\langle{#1}|}
\newcommand{\bk}[2]{\langle{#1}|{#2}\rangle}
\newcommand{\kb}[2]{\ket{#1}\bra{#2}}

\newcommand{\CC}{\mathbb{C}}
\newcommand{\RR}{\mathbb{R}}
\newcommand{\NN}{\mathbb{N}_0}

\title{Compressed Quantum Operators and Roots of Hermite Polynomials} 

\author{Serge Adonsou$^{1,2}$, Guillaume Dauphinais$^2$, David W. Kribs$^{1,2}$\footnote{Corresponding author: \href{mailto:dkribs@uoguelph.ca}{dkribs@uoguelph.ca}}, and Rajesh Pereira$^{1}$ }
\date{\small
$^1$Department of Mathematics \& Statistics, University of Guelph, Guelph, ON N1G 2W1, Canada \\
$^2$Xanadu, Toronto, ON M5G 2C8, Canada \\
\normalsize \today
}

\begin{document}

\maketitle

\begin{abstract}
The fundamental position and momentum operators of quantum mechanics are unbounded, but finite rank compressions of the operators can be considered to obtain partial information on the operators and their properties. 
Motivated by problems in photonic quantum computing, we bring together results from quantum theory and the theory of orthogonal polynomials to show that a natural finite rank compression of the position and momentum operator representation on Fock space Hilbert space has eigenvalues given by roots of the classical Hermite polynomials. 
We discuss the corresponding compressed displacement operators and potential applications in approximate quantum error correction. 
\end{abstract}


\section{Introduction}

The position and momentum operators are foundational objects of study in quantum mechanics and quantum optics. They are unbounded operators and as such they present a number of challenges with attempts to investigate them. This is of particular importance in the study of certain quantum optical systems used for reliable quantum information processing in a realistic setting, in particular using the Gottesman-Kitaev-Preskill states~\cite{gkp, IC-gkp, log-gates-GKP, BRADY2024100496, Bourassa2021blueprintscalable, ibarcq, Eickbusch, sivak, PRXQuantum.3.030301, PhysRevLett.132.150607, doi:10.1126/science.adk7560, fluhmann, deNeeve}. These are defined as eigenstates of the complex exponential of the position and momentum operators. One way to obtain partial information on problems of interest is to consider selective finite rank compressions of these operators and other related `quantum operators' in their representations on Fock space Hilbert space.  

From another direction, we have the theory of orthogonal polynomials which is a rich area of classical and modern mathematics with far reaching connections across the sciences \cite{chihara2011introduction,szeg1939orthogonal}. The Hermite polynomials are an important family of such polynomials. They arise in many settings, including in quantum theory as the eigenfunctions of the Hamiltonian of the quantum harmonic oscillator. 
These topics are clearly related, and one can find results at their intersection spread across the mathematics and physics literature, but, in our work towards developing a better understanding of error correction problems in photonic quantum computing, we have identified a need to bring together some of the key results in a novel expository self-contained presentation.  

In this paper, we clarify a connection between these quantum operators and the Hermite polynomials, with an eye toward helping to generate further progress in the infinite-dimensional quantum error correction landscape of photonic quantum computing. We show how a natural finite rank compression of the position and momentum operator representation on Fock space Hilbert space has eigenvalues given by roots of the classical Hermite polynomials. 
Our presentation has an expository emphasis, in that we carefully lay out the mathematical foundation for these operators, the orthogonal polynomials, and their relevant relations. We also show how the results yield details on quantum displacement operators that can be used in approximate quantum error correction.

This paper is organized as follows. The next section includes a brief review of the relevant background details on Hermite polynomials and the position and momentum operators. In Section~3, we review the generalized eigenvectors of the operators, and we show how the Hermite polynomials arise in their coefficients. Section~4 presents details of the result that the Hermite polynomial roots are the eigenvalues of the finite-rank compressions. We also discuss the density of the roots in that section as a consequence. In Section~5, we consider the connection with quantum error correction given by approximations of the compressed displacement operators.

\section{Background}

We begin by reviewing the relevant details of Hermite polynomials and of the position and momentum operators and observables from quantum mechanics.

\subsection{Hermite Polynomials}\label{HPF}

The class of Hermite polynomials were considered as far back as 1810 by Laplace \cite{laplace1810memoire, marquis1820theorie} and investigated by Chebyshev \cite{chebyshev1859developpement} in 1859, before ultimately being studied by their namesake Charles Hermite in 1864 \cite{hermite1864nouveau}. 
There are two ways to define the Hermite Polynomials, which are scalar multiples of each other; the `Probabilist' Hermite polynomials and `Physicist' Hermite polynomials. They are defined as follows:
\begin{Def}[Hermite Polynomials]
    The functions given by, for $n\geq 0$, 
    \begin{equation*}
        h_n(x) = (-1)^ne^{\frac{x^2}{2}}\frac{d^n}{dx^n}e^{-\frac{x^2}{2}}
    \end{equation*}
    and 
    \begin{equation*}
        H_n(x) = (-1)^ne^{x^2}\frac{d^n}{dx^n}e^{-x^2}
    \end{equation*}
    are polynomials of degree $n$, respectively called the Probabilist Hermite polynomial and the Physicist Hermite polynomial of order $n$.
\end{Def}
The Physicist Hermite polynomial $H_n$ and the Probabilist Hermite polynomial $h_n$ are degree-$n$ polynomials and are related to one another via the following relation: 
\begin{equation*}
    H_n (x ) = 2^{\frac{n}{2}}h_{n}(\sqrt{2} x).
\end{equation*}
In particular, note the probabilist version is monic. 
By induction, one obtains the following important properties of the Physicist Hermite Polynomials given in \cite{szeg1939orthogonal}.

\begin{Prop}\label{prop:rec_Hn_Hn-1}
    For all nonnegative integers $n$, we have the following recursive formulas between successive  terms:
    \begin{align*}
        H'_n(x) & = 2nH_{n-1}(x)\\
        H_{n}(x) &= 2xH_{n-1}(x)-H'_{n-1}(x) . 
    \end{align*}
\end{Prop}

The following formula is an immediate consequence of this result: 
\begin{equation*}
    H_{n+2}(x) = 2xH_{n+1}(x) -2(n+1)H_{n}(x),
\end{equation*}
which in turn implies that
\[
    h_{n+2}(x)  = xh_{n+1}(x) -(n+1)h_{n}(x).
\]

We also note that the Hermite polynomials define a family of orthogonal polynomials \cite{szeg1939orthogonal} with respect to the weight function $e^{-x^2}$; in particular, 
\begin{equation*}
    \int_{-\infty}^{\infty} e^{-x^2}H_n(x)H_m(x) dx = 2^nn!\sqrt{\pi}\, \delta_{n,m} .
\end{equation*}
One can use this relation to define the set of Hermite functions, which form an orthonormal basis for $L^2(\RR)$. 

\begin{Def}[Hermite functions]
   The sequence of functions $\phi_n$, for $n\geq 0$, defined on the real line by 
   \begin{equation}
       \phi_n(x) = \dfrac{e^{-\frac{x^2}{2}}}{\sqrt{2^nn!\sqrt{\pi}}}H_n(x) , 
   \end{equation}
   are called the Hermite functions.
\end{Def}

\subsection{Position and Momentum Operators}

We will use the following standard notation from quantum mechanics for elements from operator spectral theory: 

\strut 

\hspace{0.5in}$\ket{x}_O$ : eigenvector (ket) of observable $\hat{O}$ corresponding to eigenvalue $x$

\hspace{0.5in}$\bra{x}_O$ : dual vector (bra) associated to ket $\ket{x}_O$

\hspace{0.5in}${}_O\bk{x}{y}_{O'}$: scalar product between $\ket{x}_O$ and $\ket{y}_{O'}$

\hspace{0.5in}$\ket{x}_O\bra{y}_{O'}$: outer (operator) product

\strut 

The quantum harmonic oscillator as a quantum system is represented on the Hilbert space of square integrable complex-valued functions, 
\begin{equation*}
    L^2(\RR) = \left\{f: \RR \longrightarrow \CC: \int_{\RR} |f(x)|^2 dx < \infty\right\}.
\end{equation*}
The elements of the space give a mathematical expression of the wavefunction. The Hilbert space is separable \cite{brezis2011functional}, and as such has a countable orthonormal basis. (An example of such a countable subset is the set of  Hermite functions.) 
In particular, the separability of $L^2(\RR)$ implies that it is isometrically isomorphic to $\ell^2$,  the set of square summable series indexed by the nonnegative integers,
\begin{equation*}
\ell^2 = \left\{(\psi_n)_{n \geq 0} :\sum_{n \geq 0} |\psi_n|^2< \infty\right\}, \label{eq:l^2}
\end{equation*}
and where the isomorphism is implemented by identifying a set of basis elements for the two spaces with each other elementwise. 
We will denote the standard basis of $\ell^2$ by $\{ \ket{n} \}$ and write 
\begin{equation*}
    \ket{\psi} = \sum_{n \geq 0} \psi_n\ket{n}\label{eq:expansion_in_ketx}
\end{equation*}
for a generic element of $\ell^2$. 
As a default assumption, we will take the isomorphism as given by identifying the elements $\phi_n$ with $\ket{n}$. 

The action of observables on these states being of critical importance, we need to derive a formal expression for the position and the momentum operators on the states. Here we will also bring in the Fock basis.  
By definition, the position operator $\hat{q}$ and the momentum operator $\hat{p}$ as linear transformations on $L^2(\RR)$ are given, respectively, by
\begin{align}
    (\hat{q}\psi)(x) &= x \psi(x)\\
    (\hat{p}\psi)(x) &= -i \frac{d \psi}{d x}(x) . 
\end{align}
As we can express $\psi$ linearly with respect to the Hermite basis, we only need to express $\hat{q}\phi_n$ and $\hat{p}\phi_n$ in the Hermite basis to determine the ranges of the operators. One can verify the following equations for all $n>0$, which follow directly from the Hermite function identities: 
\[
    x\phi_n(x) = \sqrt{\frac{n}{2}}\phi_{n-1}(x) + \sqrt{\frac{n+1}{2}}\phi_{n+1}(x),\quad
    \phi_n'(x) = \sqrt{\frac{n}{2}}\phi_{n-1}(x) - \sqrt{\frac{n+1}{2}}\phi_{n+1}(x).
\]
In other words, we can express the position and momentum operators in the Fock space as follows 
\begin{align}
\hat q \ket{n} &= \sqrt{\frac{n}{2}}\ket{n-1} + \sqrt{\frac{n+1}{2}}\ket{n+1},\label{eq:xphi}\\
    \hat p \ket{n} &= \sqrt{\frac{n}{2}}\ket{n-1} - \sqrt{\frac{n+1}{2}}\ket{n+1}.\label{eq:phiprime}
\end{align}

We can also form the creation operators $\hat{a}^\dagger$ and the annihilation operators $\hat{a}$ on the Fock space, with basis also denoted for our purposes by $\{ \ket{n} \}_{n \geq 0}$. This allows us to obtain the expression of the position $\hat{q}$ and the momentum $\hat{p}$ operators. Indeed, the annihilation and creation operators are defined by their action on the Fock basis states as follows: 
\begin{align}
    \hat{a}\ket{n}& = \left\{ \begin{array}{cl} \sqrt{n} \,\, \ket{n-1} & \mbox{for}\, n\geq 1 \\ 0 & \mbox{for}\,\, n= 0, \end{array} \right. \label{op:annihilation}
\end{align}
and, 
\begin{align}
    \hat{a}^\dagger \ket{n}& = \sqrt{n+1}\, \ket{n+1} \,\,\,\, \mbox{for}\,\, n\geq 0.\label{op:creation}
\end{align}

Using these two expressions, we can write the position and momentum as expressions of the annihilation and creation operators as follows, 
\begin{align}
    \hat{q} &= \frac{1}{\sqrt{2}}\left(\hat{a}+\hat{a}^\dag\right),\label{op:posit_fct_ac}\\
    \hat{p} &= \frac{i}{\sqrt{2}}\left(-\hat{a}+\hat{a}^\dag\right).\label{op:moment_fct_ac}
\end{align}
The well-known commutation relationship between these two operators is then found by straightforward computation and given by: 
\begin{equation*}
    [\hat{q},\hat{p}] = iI,
\end{equation*}
which correspondingly gives the commutation relation between the annihilation and creation operators;  
\begin{equation*}
    [\hat{a}, \hat{a}^\dagger] = I.
\end{equation*}

Several other important operators in quantum mechanics are similarly determined by the position and momentum operators. For instance, the so-called number operator is diagonalized in the Fock basis and given by
\begin{equation*}
    \hat{n} = \hat{a}^\dagger\hat{a}=\frac{1}{2}\left(\hat{q}^2+\hat{p}^2-I\right) = \hat{H} -\frac{1}{2}, 
\end{equation*}
where $\hat{H} = (1/2)(\hat{q}^2+\hat{p}^2)$ is the Hamiltonian for the system. Further, the Fourier transform is given by 
\begin{equation*}
    F = \sum_{n\geq 0} (-i)^n\kb{n}{n} . 
\end{equation*}

\section{Hermite Polynomials and Generalized Eigenvectors}

The position and momentum operators do not have convergent eigenvectors, but nevertheless the formal (divergent) forms are expressible in terms of the Hermite polynomials. We review this folklore-type result in this section as it naturally leads into what follows. 

\subsection{`Eigenvalues' and `Eigenvectors' of the Position and Momentum Operators}

We will focus on the position operator for this derivation. 
We begin by noting from the last section that the position operator acts on basis elements (identified with the Hermite functions) of $L^2(\mathbb{R})$ by,    
\begin{align}
    \hat{q}\ket{0}&= \frac{1}{\sqrt{2}} \ket{1}\label{eq:q_on_0}\\ 
    \hat{q}\ket{n}& = \frac{1}{\sqrt{2}}\sqrt{n}\ket{n-1}+\frac{1}{\sqrt{2}}\sqrt{n+1}\ket{n+1},\quad \mbox{for}\,\, n\geq 1. \label{eq:q_on_n}
\end{align}
It is well known that the position operator, even being self-adjoint, does not have eigenvectors in the space of square-integrable functions. Let us investigate what form such eigenvectors would have, independent of convergence considerations. 

We start with the eigenvector-eigenvalue equation given by 
$    \hat{q} \ket{\psi}= \lambda \ket{\psi}, $
for a real number $\lambda$ and non-zero vector state $\ket{\psi}$ that we assume has decomposition in terms of the basis states given by  
$\ket{\psi} = \sum_{n\geq 0} x_n \ket{n}.$
We can then rewrite this equation as: 
\[
    \hat{q}\left( \sum_{n\geq 0} x_n \ket{n}\right) = 
    \frac{x_1}{\sqrt{2}}\ket{0} +\sum_{n\geq 1} \frac{1}{\sqrt{2}}\left(x_{n+1}\sqrt{n+1}+x_{n-1}\sqrt{n}\right)\ket{n} = \sum_{n\geq 0} \lambda x_n \ket{n}.
\]
By equating coefficients corresponding to the same basis elements, we obtain the following equations: 
\begin{align*}
    \frac{x_1}{\sqrt{2}} & = \lambda x_0,\\
    \frac{1}{\sqrt{2}}\left(x_{n+1}\sqrt{n+1}+x_{n-1}\sqrt{n}\right)& = \lambda x_n, \quad \mbox{for} \,\, n\geq 1. 
\end{align*}
We perform a change of variable to get $\delta = \sqrt{2}\lambda$, and this gives:
\begin{align*}
    x_1 & = \delta x_0,\\
    x_{n+1}\sqrt{n+1}+x_{n-1}\sqrt{n}& = \delta x_n, \quad \mbox{for} \,\, n\geq 1.
\end{align*}

We claim that 
\begin{equation}
    x_n\sqrt{n!} = h_n(\delta) x_0
\end{equation}
where $h_n$ is the Probabilistic Hermite polynomial. We can prove this claim using induction.
\begin{itemize}
    \item Base case: for $n =1$, we get $x_1 = \delta x_0$ as expected. 
    \item Induction: Assuming that for all $k\leq n$, $x_k\sqrt{k!} = h_k(\delta)x_0$, let us show that $x_{n+1}\sqrt{n+1!} = h_{n+1}(\delta) x_0$. Indeed, we have
    \begin{equation}
        x_{n+1}\sqrt{n+1}+x_{n-1}\sqrt{n} = \delta x_n
    \end{equation}
    By multiplying both side by $\sqrt{n!}$, we get
    \begin{align*}
        x_{n+1}\sqrt{n+1!}+x_{n-1}n\sqrt{n-1!} &= \delta x_n\sqrt{n!},\nonumber\\
        x_{n+1}\sqrt{n+1!}+nh_{n-1}(\delta)x_0&= \delta h_n(\delta)x_0,\nonumber\\
        x_{n+1}\sqrt{n+1!} &= \left[\delta h_n(\delta)-nh_{n-1}(\delta)\right]x_0,\\ 
        x_{n+1}\sqrt{n+1!} &= h_{n+1}(\delta)x_0.
    \end{align*}
    The last equality is obtained from the recurrence formula verified by the Probabilist Hermite polynomial. This proves that the claim is true. 

\end{itemize}

Now that the claim is proved, we can also rewrite it in terms of the Physicist Hermite polynomial, which gives
\begin{equation*}
    x_n = \frac{1}{\sqrt{2^nn!}}H_n(\lambda)x_0.
\end{equation*}
This means that for given eigenvalue $\lambda$, the associated (formal) eigenvector $\ket{\lambda}_q$ satisfies:  
\begin{equation*}
    \ket{\lambda}_q := \sum_{n\geq 0}\frac{1}{ \sqrt{2^nn!}}H_n(\lambda) \ket{n}.\label{eq:q_eigen_serie_expression}
\end{equation*}
We claim that for all real numbers $\lambda$, the series \[\sum_{n\geq 0}\frac{1}{2^nn!}H_n(\lambda)^2 =  \infty.\]
Indeed, in general for all positive numbers $r$, we have from Mehler's Theorem that 
\begin{equation*}
    \sum_{n\geq 0}\frac{r^n}{2^nn!}H_n(\lambda)^2 = \frac{1}{\sqrt{1-r^2}}\exp\left(\frac{r^2-r}{r-1}2\lambda^2\right).
\end{equation*}
Note that the limit of the exponential on the right-hand side when $r\to 1$ is $\exp(2\lambda^2)$ but $(1-r^2)^{-\frac{1}{2}}$ tends to infinity. Thus, using the monotone convergence theorem, one can permute the limit with respect to $r\to 1$ and the summation, which justifies the divergence claimed. Therefore, the position operator has no eigenvectors belonging to the space of square summable sequences. 

The same is true for the momentum operator, with the corresponding equation given by: 
\begin{equation*}
 \hat{p}\left( \sum_{n\geq 0} x_n \ket{n}\right) =    -\frac{ix_1}{\sqrt{2}}\ket{0} +i\sum_{n\geq 1} \frac{1}{\sqrt{2}}\left(-x_{n+1}\sqrt{n+1}+x_{n-1}\sqrt{n}\right)\ket{n} = \sum_{n\geq 0} \lambda x_n \ket{n}.
\end{equation*}
Thus, in that case we get
\begin{align*}
    \frac{x_1}{\sqrt{2}} & = i\lambda x_0,\\
    \frac{1}{\sqrt{2}}\left(x_{n+1}\sqrt{n+1}-x_{n-1}\sqrt{n}\right)& = i\lambda x_n, \quad \mbox{for}\,\, n\geq 1. 
\end{align*}
One can prove by induction as above that for all $n\geq 0$, we have 
\begin{equation*}
    x_n = \frac{i^n}{\sqrt{2^nn!}}H_n(\lambda)x_0.
\end{equation*}
Hence, if the eigenvectors of the momentum operator associated to the eigenvalue $\lambda$ exist, they would be given by:
\begin{equation*}
    \ket{\lambda}_p := \sum_{n\geq 0}\frac{i^n}{ \sqrt{2^nn!}}H_n(\lambda) \ket{n}.
\end{equation*}

It follows that these sums do not converge. 

\begin{Rm} 
In fact, formally $\ket{\lambda}_q$ and $\ket{\lambda}_p$ reside in a space containing the Hilbert space of square-summable sequences, defined through their actions as linear functionals. This space is denoted by $s^\times$, the space of sequences bounded by a polynomial sequence (i.e., slow increasing sequences), which is identified as the antidual of $s$, the space (called Schwartz space) of rapidly decreasing sequences (see \cite{reedsimonmethods}). More specifically, we have $s\subset l^2\subset s^\times$, known as a  Gelfand triplet, with 
\begin{equation}
    s = \left\{\ket{\psi}\in \CC^{\NN}: \sum_{n\geq 0}(n+1)^{2m}|\psi_n|^2< \infty\mbox{ for all } m\in \NN\right\}. \label{set:SN}
\end{equation}
\end{Rm}

\subsection{Partial Sums of the `Eigenvectors'}

We have seen in the previous section that the series giving the eigenvectors of position and momentum operators does not converge. Here we consider their partial sums, which will give us some motivation for what follows. We have the following lemma that tells us how much the partial summation from \Cref{eq:q_eigen_serie_expression} deviates in the eigenvector-eigenvalue equation.

\begin{Lem}\label{lem:q_on_partial_lambda_k}
    For all integers $k \geq 1$, we have
    \begin{equation}
        \hat{q}\left(\sum_{n = 0}^{k-1}\frac{1}{ \sqrt{2^nn!}}H_n(\lambda) \ket{n}\right) = \lambda \sum_{n = 0}^{k-1}\frac{1}{ \sqrt{2^nn!}}H_n(\lambda) \ket{n}+ \frac{H_k(\lambda)}{2\sqrt{2^{k-1}(k-1)!}}\ket{k-1} + \frac{k}{ \sqrt{2^{k}k!}}H_{k-1}(\lambda) \ket{k}.
    \end{equation}
\end{Lem}

\begin{proof}
    Using \Cref{eq:q_on_0,eq:q_on_n}, we can write that 
    \begin{equation*}
        \hat{q}\left(\sum_{n = 0}^{k-1}\frac{1}{ \sqrt{2^nn!}}H_n(\lambda) \ket{n}\right) = \ket{1} + \sum_{n = 1}^{k-1}\frac{1}{ \sqrt{2^nn!}}H_n(\lambda) \left(\sqrt{\frac{n}{2}}\ket{n-1}+\sqrt{\frac{n+1}{2}}\ket{n+1}\right).
    \end{equation*}
    We split the summation, shift the index forward and backward as needed to get: 
    \begin{align*}
        \hat{q}\left(\sum_{n = 0}^{k-1}\frac{1}{ \sqrt{2^nn!}}H_n(\lambda) \ket{n}\right)& = \ket{1} + \sum_{n = 0}^{k-2}\frac{1}{ \sqrt{2^nn!}}\frac{H_n(\lambda)}{2} \ket{n}\\
        & \quad\quad\,\,\, + \sum_{n = 2}^{k}\frac{1}{ \sqrt{2^nn!}}nH_{n-1}(\lambda) \ket{n}.
    \end{align*}
    We use the recurrence formula characterizing the Hermite polynomial to get:  
    \begin{align*}
        \hat{q}\left(\sum_{n = 0}^{k-1}\frac{1}{ \sqrt{2^nn!}}H_n(\lambda) \ket{n}\right) & = \lambda \sum_{n = 0}^{k-2}\frac{1}{ \sqrt{2^nn!}}H_n(\lambda) \ket{n}+ \sum_{n = k-1}^{k}\frac{n}{ \sqrt{2^{n}n!}}H_{n-1}(\lambda) \ket{n} \\
        & = \lambda \sum_{n = 0}^{k-1}\frac{1}{ \sqrt{2^nn!}}H_n(\lambda) \ket{n}+ \frac{H_k(\lambda)}{2\sqrt{2^{k-1}(k-1)!}}\ket{k-1} +\frac{k}{ \sqrt{2^{k}k!}}H_{k-1}(\lambda) \ket{k}.
    \end{align*}
    The last equation is obtained using the recurrence formula that expresses $H_{k-2}$ in terms of $H_{k-1}$ and $H_{k}$.   
\end{proof}


As the above lemma determines the action of the position operator, using the fact that the momentum and position operators are Fourier-conjugate, we can express the action of momentum on the partial sum of a generalized eigenvector of the momentum operator. We get the following equation:  
\begin{equation*}
        \hat{p}\left(\sum_{n = 0}^{k-1}\frac{i^n}{ \sqrt{2^nn!}}H_n(\lambda) \ket{n}\right) = \lambda \sum_{n = 0}^{k-1}\frac{i^n}{ \sqrt{2^nn!}}H_n(\lambda) \ket{n}+ \frac{i^{k-1}H_k(\lambda)}{2\sqrt{2^{k-1}(k-1)!}}\ket{k-1} + \frac{i^kk}{ \sqrt{2^{k}k!}}H_{k-1}(\lambda) \ket{k}.
    \end{equation*}


We have shown that the action of the position and momentum operators on the partial sum of their generalized eigenvectors hints at an analysis of compressed versions of these operators to particular finite-dimensional subspaces defined by the Fock basis. We turn to this in the next section.

\section{Compressed Operators and Roots of Hermite Polynomials}

\subsection{Compressed Position and Momentum Operators}

We consider the following family of finite-dimensional subspaces $\{\mathcal V_k\}_{k\geq 0}$, where each $\mathcal V_k$ is the $k$-dimensional space generated by the first $k$-Fock basis states; that is,
\begin{equation*}
    \mathcal V_k = \left\{\alpha_0\ket{0}+\dots +\alpha_{k-1}\ket{k-1}:(\alpha_0, \dots, \alpha_{k-1})\in \CC^k\right\}.
\end{equation*}
These spaces are finite-dimensional truncations of the Fock space. Physically, this means that for a general system consisting of arbitrary particles, one approximates the state space as the superposition of states with the number of particles not larger than some positive integer. The norm-closure of the increasing union $\cup_{k\geq 0} \mathcal V_k$ is the entire Hilbert space.

Let $P_k$, for $k\geq 0$, be the orthogonal projection of the Fock space onto $\mathcal V_k$. Then on basis states, we have  
\begin{equation*}
    P_k\ket{n} = \begin{cases}
        \ket{n}& \mbox{ if } n\leq k-1,\\
        0& \mbox{ if } n>k-1.
    \end{cases}
\end{equation*}
Note that the sequence of projections $P_k$ converge as $k\rightarrow \infty$ to the (infinite rank) identity operator $I$ in the {\it strong operator topology}, which is described by pointwise convergence; that is, for all $\ket{x}\in\ell^2$, 
\[
  \|\ket{x} - P_k\ket{x}\|_2^2  = \left\|\sum_{n\geq k}x_n\ket{n}\right\|_2^2 = \sum_{n\geq k}|x_n|^2\stackrel[k\to \infty]{}{\longrightarrow} 0 . 
\]
This happens even though each $P_k$ is finite rank and the projection $P_k^\perp$ onto its orthogonal complement $\mathcal V_k^\perp$ is always infinite rank. 

More generally, for any operator $A$ acting on $\ell_2$, the sequence of compressed operators $A_k = P_kA P_k$ converges similarly to $A$.

\begin{Def}
For every operator $A$ acting on the space $\ell_2$, we define $A_k$ to be its compression to the $k$-dimensional space $\mathcal V_k$; that is,  
\begin{equation*}
    A_k   = P_kA P_k.\label{eq:compression}
\end{equation*}
\end{Def}

The compression $A_k$ that acts on the extended space $\ell_2$ is finite-rank as its image lies in the $k$-dimensional subspace $\mathcal V_k$. Moreover, in general, even if $A$ is an unbounded operator, $AP_k$ is bounded, so that $A_k$ is a continuous operator. The restriction of $A_k$ to the space $\mathcal V_k$ defines an endomorphism on $\mathcal V_k$, and can be identified with a square matrix through standard matrix representation theory \cite{moretti2017spectral}.

From this perspective, the `$k$-compressed' versions of the annihilation and creation operators $a_k$ and $a^\dag_k$ satisfy
\begin{equation*}
    a_k\ket{n} = \begin{cases}
        \sqrt{n}\ket{n-1} & \mbox{ if } n\leq k-1 \\
        0& \mbox{ if } n>k-1,
    \end{cases}
\end{equation*}
and 
\begin{equation*}
    a^\dagger_k\ket{n} = \begin{cases}
        \sqrt{n+1}\ket{n+1} & \mbox{ if } n\leq k-2 \\
        0& \mbox{ if } n>k-2.
    \end{cases}
\end{equation*}
Following directly from these identities and the general relations between $a$, $a^\dagger$ and $\hat{q}$, $\hat{p}$, we get the actions of the $k$-compressed position and momentum operators as follows (here using the convention $\ket{-1}=0$): 
\begin{align*}
    \hat{q}_k\ket{n} &= \begin{cases}
        \frac{1}{\sqrt{2}}\sqrt{n}\ket{n-1}+\frac{1}{\sqrt{2}}\sqrt{n+1}\ket{n+1}&\mbox{ if } n\leq k-2 \\
        \frac{1}{\sqrt{2}}\sqrt{n}\ket{n-1}&\mbox{ if } n =  k-1 \\
        0& \mbox{ if } n>k-1 
    \end{cases}\\
    \hat{p}_k\ket{n} &= \begin{cases}
        -\frac{i}{\sqrt{2}}\sqrt{n}\ket{n-1}+\frac{i}{\sqrt{2}}\sqrt{n+1}\ket{n+1}&\mbox{ if } n\leq k-2 \\
        -\frac{i}{\sqrt{2}}\sqrt{n}\ket{n-1}&\mbox{ if } n =  k-1 \\
        0& \mbox{ if } n>k-1
    \end{cases}
\end{align*}
From these equations, one can write the action of the compressed position and compressed momentum on any element $\ket{x}= \sum_{n\geq 0} x_n \ket{n}$ of $\ell_2$ as follows: 
\begin{equation}
    \hat{q}_k \ket{x} =\frac{x_1}{\sqrt{2}}\ket{0} +\sum_{n=1}^{k-2} \frac{1}{\sqrt{2}}\left(x_{n+1}\sqrt{n+1}+x_{n-1}\sqrt{n}\right)\ket{n}+\frac{x_{k-2}}{\sqrt{2}}\sqrt{k-1}\ket{k-1}\label{prop:eq2}
\end{equation}
\begin{equation}
         \hat{p}_k \ket{x} =-\frac{ix_1}{\sqrt{2}}\ket{0} -\sum_{n=1}^{k-2} \frac{i}{\sqrt{2}}\left(x_{n+1}\sqrt{n+1}-x_{n-1}\sqrt{n}\right)\ket{n}+i\frac{x_{k-2}}{\sqrt{2}}\sqrt{k-1}\ket{k-1}.\label{prop:eq3}
\end{equation}

The commutator relations between the compressed operators are then found by direct calculation in the next result.  

\begin{Prop}\label{compressedcommutator}
    For $k \geq 1$, the commutation relationship between the $k$-compressed position and momentum operators satisfies:  
\begin{equation}\label{eq:commutation_qkpk_simply}
     [\hat{q}_k, \hat{p}_k] = i P_{k-1} -i(k-1)\kb{k-1}{k-1}.
\end{equation}
Equivalently, we can write this relation in terms of the full position and momentum operators as follows, 
\begin{equation}\label{eq:commutation_qkpk}
    [\hat{q}_k, \hat{p}_k] = [\hat{q}, \hat{p}] -ik\kb{k-1}{k-1}-i P_k^\perp . 
\end{equation} 
\end{Prop}

\begin{proof}
We have $[\hat{q}_k,\hat{p}_k]\ket{n} =[\hat{q}, \hat{p}]\ket{n}$ for all $n\leq k-2$, and $[\hat{q}_k,\hat{p}_k]\ket{n} = 0$ for all $n\geq k$. Whereas, on the $(k-1)$-th Fock basis state, we have 
    \begin{align*}
        [q_k,p_k]\ket{k-1}& =(\hat{q}_k\hat{p}_k - \hat{p}_k\hat{q}_k)\ket{k-1},\nonumber\\
         &= -i\frac{\sqrt{k-1}}{\sqrt{2}}\hat{q}_k\ket{k-2}  +  \frac{\sqrt{k-1}}{\sqrt{2}}\hat{p}_k\ket{k-2}, \nonumber\\
        & = -i\frac{\sqrt{k-1}}{\sqrt{2}}\left(\frac{\sqrt{k-2}}{\sqrt{2}}\ket{k-3} - \frac{\sqrt{k-1}}{\sqrt{2}}\ket{k-1}\right)  +  i\frac{\sqrt{k-1}}{\sqrt{2}}\left(\frac{\sqrt{k-2}}{\sqrt{2}}\ket{k-3} - \frac{\sqrt{k-1}}{\sqrt{2}}\ket{k-1}\right), \nonumber\\
        & = -i(k-1)\ket{k-1}.
    \end{align*}

   This yields \Cref{eq:commutation_qkpk_simply}. To obtain \Cref{eq:commutation_qkpk}, we simply use the commutator $ [\hat{q},\hat{p}] = iI$ of $\hat{q}$ and $\hat{p}$ and rearrange the equation accordingly.
\end{proof}

\subsection{Diagonalization of the Compressed Operators with Hermite Polynomial Roots}

The position and momentum operators are self-adjoint unbounded operators on $\ell_2$. However, their $k$-compressed versions are self-adjoint and finite-rank on $\ell_2$. Hence the spectral theorem applies and these compressed versions are diagonalizable. We identify the elements of their spectra in the following result.

\begin{Thm}\label{mainlem}
    For all $k>1$, every non-zero eigenvalue of the $k$-compressed position $\hat{q}_k$ and momentum $\hat{p}_k$ operators is a root of the $k$-th Hermite polynomial $H_k(x)$. 
\end{Thm}

We shall first show this directly and then we will indicate below how the result can also be derived from a result in the theory of orthogonal polynomials. 

\begin{proof}
For all $k>1$, we can use the equations above and consider the matrix representations for the compressed operators in the (ordered) orthonormal basis $\{\ket{n}\}_{n\geq 0}$. In matrix form, the relevant operator eigenvalue equation gives (with all unwritten entries equal to zero),  
\begin{equation*}
    \hat{q}_k-\lambda I =\left[
    \begin{array}{rrrcl|rccc}
-\lambda  & \frac{1}{\sqrt{2}}            &          &              &            & & & \\
        \frac{1}{\sqrt{2}} & -\lambda     & \frac{1}{\sqrt{2}}\sqrt{2} &              &            & & & \\
          & \frac{1}{\sqrt{2}}\sqrt{2}     & -\lambda & \ddots       &            & & & \\
          &              & \ddots   & \ddots       &\frac{1}{\sqrt{2}}\sqrt{k-1}  & & & \\
          &              &          &  \frac{1}{\sqrt{2}}\sqrt{k-1}  & -\lambda   & & & \\
\hline
          &             &           &              &            &-\lambda & 0 & \\
          &             &           &              &            & 0 & -\lambda  & \ddots \\
          &             &           &              &            & & \ddots &\ddots   \\
    \end{array}\right].
\end{equation*}

In particular, a non-zero eigenvalue of this matrix must be an eigenvalue of the left upper block matrix. To find these eigenvalues, we examine the roots of the characteristic polynomial of that matrix, which we will denote by $q_k(\lambda) = | (\hat{q}_k - \lambda I)|_{\mathcal V_k} |$. 

For the base cases, $k=2,3$, we have $q_2(\lambda) = \lambda^2 - \frac12$ and $q_3(\lambda) = -\lambda^3 + \frac32\lambda$, and so $q_k(\lambda) =0$ if and only if $H_k(\lambda) =0$, as $H_2(x) = 4 x^2 - 2$ and $H_3(x) = 8x^3 -12x$. Note also in these cases that $H_k(x) = 2^k q_k(-x)$. 

For $k>3$, note that $q_k(\lambda) = \sqrt{2}^{-k} g_k(-\sqrt{2}{\lambda})$ where, 
\begin{align*}
g_k(\lambda)= & \left|
    \begin{array}{rrrcl}
\lambda  &      1        &          &              &            \\
      1    & \lambda     &  \sqrt{2} &              &            \\
          &  \sqrt{2}     & \lambda & \ddots       &            \\
          &              & \ddots   & \ddots       & \sqrt{k-1}  \\
          &              &          &   \sqrt{k-1}  & \lambda   
     \end{array}\right|  \\  = & {\lambda \left|
    \begin{array}{rrrcl}
\lambda  &      1        &          &              &            \\
     1     & \lambda     &  \sqrt{2} &              &            \\
          &  \sqrt{2}     & \lambda & \ddots       &            \\
          &              & \ddots   & \ddots       & \sqrt{k-2}  \\
          &              &          &   \sqrt{k-2}  & \lambda  
     \end{array}\right|}
     - \sqrt{k-1} \left|
    \begin{array}{rrrccl}
\lambda  &      1        &          &             &           &            \\
    1      & \lambda     &  \sqrt{2} &            &           &            \\
          & \sqrt{2}     & \lambda & \ddots       &           &            \\
          &              & \ddots   & \ddots      &  \sqrt{k-3} &          \\
          &              &          &  \sqrt{k-3} & \lambda  &  \sqrt{k-2} \\
          &              &          &             &          &  \sqrt{k-1}   
     \end{array}\right|.
\end{align*}
In the determinant calculation above, we have expanded along the final column. Therefore, we evidently have $g_k(\lambda) =  \lambda g_{k-1}(\lambda) -(k-1) g_{k-2}(\lambda)$, which is precisely the recurrence formula for the Probabilist Hermite polynomials. 
Recalling the relation $H_k(\lambda) = \sqrt{2}^{k}h_k(\sqrt{2}\lambda)$, we have $H_k(\lambda) = \sqrt{2}^{k}g_k(\sqrt{2}\lambda)$, and so $H_k(\lambda) = 2^kq_k(-\lambda)$ for all $k$. 

Given that the eigenvalues we desire are the zeros of $q_k$, it follows that these eigenvalues are the symmetric opposites on the real line of the zeros of the Hermite polynomials, which are symmetric themselves. Hence, these eigenvalues are exactly the zeros of the Hermite polynomial $H_k$. 

Similarly, for the momentum operators $\hat{p}_k$, we can define the polynomials $g_k(\lambda)$ as above, and we obtain the relation $g_k(\lambda) =  \lambda g_{k-1}(\lambda) -(k-1) g_{k-2}(\lambda)$, exactly as above.   
\end{proof}

\begin{Rm}
This result has appeared in the literature, for instance see \cite{trzetrzelewski2003quantum}.
In fact, we have found the underlying mathematical basis of this result is a special case of a general result that every $n$th degree polynomial in a sequence of orthogonal polynomials is the characteristic polynomial of a certain self-adjoint tridiagonal matrix. The following result can be found in \cite[pg. 77]{borwein2012polynomials} and \cite[pg. 30]{chihara2011introduction}. One can see a special case of the tridiagonal matrix form in the result above.  
\end{Rm} 

\begin{Thm}
Let  $\{a_n\}_{n=0}^{\infty}$ be a sequence of real numbers and  $\{b_n\}_{n=0}^{\infty}$ be a sequence of positive numbers, and let $\{f_n\}_{n=0}^{\infty}$ be the sequence of polynomials generated by the three term recurrence $f_{n+1}(x)=(x-a_n)f_n(x)-b_n f_{n-1}(x)$ with $f_{0}=1$ and $f_1=(x-a_0)$.  Then for all $k\in \mathbb{N}$, $f_k$ is the characteristic polynomial of the $k \times k$ self-adjoint tridiagonal matrix (with unwritten entries equal to zero):
$$\left[
    \begin{array}{rrrccl}
  a_0 &      \sqrt{b_1}       &          &             &           &            \\
    \sqrt{b_1}     & a_1    & \sqrt{b_2} &            &           &            \\
          & \sqrt{b_2}     & a_2 & \ddots       &           &            \\
    &              & \ddots   & \ddots       & \sqrt{b_{k-1}}  \\
          &              &          &   \sqrt{b_{k-1}}  & a_{k-1}
            
     \end{array}\right].
     $$
\end{Thm}

We can relate the description of eigenvalues in the spectral decomposition of the compressed operators above to the formal description of the full position and momentum operators. 
In particular, the next result expresses the eigenvectors of the compressed position operators and the uncompressed position operator when these two are associated with the same eigenvalue, namely zeros of Hermite polynomials. 
Using standard conventions from quantum mechanics, we denote the eigenvectors associated to an element $\lambda$ in the spectrum of $\hat{q}_k$, $\hat{p}_k$ by $\ket{\lambda}_{q_k}$ and $\ket{\lambda}_{p_k}$, and so we have the diagonal forms   
\[
        \hat{q}_k = \sum_{j = 0}^{k-1} \lambda_{k,j}\ket{\lambda_{k,j}}_{q_k}\bra{\lambda_{k,j}}_{q_k}
          \quad \mathrm{and} \quad 
        \hat{p}_k = \sum_{j = 0}^{k-1} \lambda_{k,j}\ket{\lambda_{k,j}}_{p_k}\bra{\lambda_{k,j}}_{p_k}, 
\]
    where $\{\lambda_{k,j}\}_{j=0}^{k-1}$ are the $k$ roots of the Hermite polynomial of  degree $k$, and the lists of vectors $\ket{\lambda_{k,0}}_{q_k}$,$\ket{\lambda_{k,1}}_{q_k}$, $\dots$, $\ket{\lambda_{k,k-1}}_{q_k}$ and  $\ket{\lambda_{k,0}}_{p_k}$,$\ket{\lambda_{k,1}}_{p_k}$, $\dots$, $\ket{\lambda_{k,k-1}}_{p_k}$ each form an orthonormal basis of $\mathcal V_k$.

We note in the statement of the following result that although each projection $P_k$ is defined on the Hilbert space $l^2$, and that $\ket{\lambda}_q$ is not an element of the Hilbert space, we can extend the definition of  $P_k$ to the space $s^\times$ containing $\ket{\lambda}_q$  by continuity. 

\begin{Coro} \label{Thm:eigenvectors}
The compressed position operator $\hat{q}_k$ and the compressed momentum operator $\hat{p}_k$ are diagonalizable with eigenvalues $\{\lambda_{k,j}\}_{j=0}^{k-1}$ zeros of the Hermite polynomial of degree $k$ and the associated eigenvectors satisfy 
\begin{align}
        \ket{\lambda_{k,j}}_{q_k} &=  P_{k}\ket{\lambda_{k,j}}_{q} \label{eq:eigenvector_qk_in_k-1}\\
        \ket{\lambda_{k,j}}_{p_k} &=  P_{k}\ket{\lambda_{k,j}}_{p} \label{eq:eigenvector_pk_in_k-1}
    \end{align}
    where $\ket{\lambda_{k,j}}_{q}$ and $\ket{\lambda_{k,j}}_{p}$ are respectively the generalized eigenvectors of the uncompressed position and momentum operators for $\lambda_{k,j}$. Moreover, the Fourier transform relates the eigenvectors by, 
    \begin{equation}
        \ket{\lambda_{k,j}}_{p_k} = F^*\ket{\lambda_{k,j}}_{q_k}.\label{eq:pk_qk}
    \end{equation}
\end{Coro}

\begin{proof}
    From \Cref{lem:q_on_partial_lambda_k}, for any real number $\lambda$,  we have 
    \begin{equation}\label{eq:pi_k_p_partial_k}
         P_k \hat{q}\left(\sum_{n = 0}^{k-1}\frac{1}{ \sqrt{2^nn!}}H_n(\lambda) \ket{n}\right) = \lambda \sum_{n = 0}^{k-1}\frac{1}{ \sqrt{2^nn!}}H_n(\lambda) \ket{n}+ \frac{H_k(\lambda)}{2\sqrt{2^{k-1}(k-1)!}}\ket{k-1}.
    \end{equation}
    We denote
    \begin{equation*}
         P_k\ket{\lambda}_q = \sum_{n = 0}^{k-1}\frac{1}{ \sqrt{2^nn!}}H_n(\lambda) \ket{n}.
    \end{equation*}
    This allows us to rewrite \Cref{eq:pi_k_p_partial_k} as follows:
    \begin{equation*}
        \hat{q}_k P_k\ket{\lambda}_q = \lambda P_k\ket{\lambda}_q + \frac{H_k(\lambda)}{2\sqrt{2^{k-1}(k-1)!}}\ket{k-1}.
    \end{equation*}
    Now we pick $\lambda$ to be a zero of the $k$-th Hermite polynomial. This implies that $H_k(\lambda) =0$, and we get
    \begin{equation*}
        \hat{q}_k P_k\ket{\lambda}_q = \lambda P_k\ket{\lambda}_q .
    \end{equation*}
    Similarly, 
    \begin{equation*}
        \hat{p}_k P_k\ket{\lambda}_p = \lambda P_k\ket{\lambda}_p .
    \end{equation*}

We obtain \Cref{eq:pk_qk} directly from the definition of the Fourier transform, which acts as a change of basis matrix between the operators. 
 \end{proof}

We further note that the change of basis matrix between the bases $\{\ket{\lambda_{j',k}}_{q_k}\}_{j'}$ and $\{\ket{\lambda_{j,k}}_{p_k}\}_j$ has as entries $ {}_{q_k}\bk{\lambda_{k,j}}{\lambda_{k,j'}}_{p_k} =I_k(\lambda_{k,j}, \lambda_{k,j'})$, such that 
\begin{equation}\label{eq:changeofbasis}
    I_k(x,y) = \sum_{n =0}^{k-1}\frac{i^n}{2^nn!}H_n(x)H_n(y).
\end{equation}
The series of two-variable functions $I_k$ converges to the Melher kernel function.

\subsection{Density of the Roots}

An important result from the theory of orthogonal polynomials is that the set of zeros of the Hermite polynomials forms a dense subset of the real line. Indeed, to see this, it is known  \cite{dominici2010asymptotic} that for all nonnegative integers $n\geq 0$, the absolute value of any root of $H_n$ does not exceed $\sqrt{2n}$. In fact, this bound of $\sqrt{2n}$ can be derived by a simple application of Gershgorin's theorem to the tridiagonal matrix representation of $\hat{q}_n$ or $\hat{p}_n$ in the Fock basis. Hence, the zeros of $H_n$ with the greatest magnitude are of order $o(n)$. These facts are used as ingredients for Theorem 6.1.2 in \cite{szeg1939orthogonal}, which is given as follows.

\begin{Thm}[\cite{szeg1939orthogonal} ]
    Suppose $(f_n)$ is a set of orthogonal polynomials with weight function $W(x)$ supported on $\mathbb R$ such that the zero of greatest modulus of $f_n(x)$ is $o(n)$. For every closed interval $[a, b]$, such that $\int_{a}^{b}W(x)dx> 0$, there is $n_0$ such that for all $n>n_0$, every polynomial $f_n$ has at least one zero in $[a,b]$. 
\end{Thm}

It follows that every non-empty closed interval of the real line contains Hermite polynomial zeros.     
As a consequence of the result above, we can give an alternate proof of this classical result using well-known properties of the position and momentum operators in quantum theory. In this approach, the density of the zeros of the Hermite polynomial on the real line emerges as a consequence of the convergence of the compressed position/momentum operators to the non-compressed position/momentum operator. Indeed, we consider the sequence of compressed position operators $\{\hat{q}_k\}$ as acting on the Hilbert space $l^2$ as given in  Equation~\ref{prop:eq2}  and $\hat{q}$ as given in Equation~\ref{eq:xphi}. All these operators are self-adjoint operators with a common core  $S(\NN)$. Additionally, for all $\ket{\psi} = (\psi_0, \psi_1, \dots, \psi_k, \dots)\in S(\NN)$, the sequence $\{(k-1)\psi_{k-2}^2\}$ converges to zero; therefore, 
\[ 
    \frac{\psi_1^2}{2} +\sum_{n=1}^{k-2} \frac{1}{2}\left(\psi_{n+1}\sqrt{n+1}+\psi_{n-1}\sqrt{n}\right)^2+\frac{\psi_{k-2}^2}{2}(k-1) \,\, \stackrel{k \rightarrow \infty}{\longrightarrow} \,\,  \frac{\psi_1^2}{2} +\sum_{n=1}^{\infty} \frac{1}{2}\left(\psi_{n+1}\sqrt{n+1}+\psi_{n-1}\sqrt{n}\right)^2.
\]
This means that for all $\ket{\psi}\in S(\NN)$, $\hat{q}_k\ket{\psi} \to \hat{q}\ket{\psi}$ in $l_2$; this tells us that the sequence of self-adjoint operators $\hat{q}_k$ converges to $\hat{q}$ in the strong resolvent sense (as defined in \cite[pg. 284]{reedsimonmethods}). This justifies the well-known result:

\begin{Coro}
    The zeros of the Hermite polynomial form a dense subset of the real line. 
\end{Coro}

\section{Displacement Operators and Approximate Quantum Error Correction}

We can also use the compressed operators to construct sequences of sets of unitary operators, preserving the compressed subspaces, that approximate the {\it displacement operators} which are fundamental in quantum optics.

\begin{Def}[Discrete-spectrum displacement operators]
For any positive integer $k$ and $\alpha, \beta\in \mathbb{R}$, we call the unitary operators $U_{k,\alpha} = \exp(i\alpha\hat{q}_k)$ and $V_{k,\beta} = \exp(-i\beta\hat{p}_k)$ the discrete-spectrum displacement operators. 
\end{Def}

The unitary operators $U_{k,\alpha}$, $V_{k,\beta}$ are defined via the (self-adjoint) functional calculus for $\hat{q}_k$, $\hat{p}_k$. Their explicit actions are given as follows. 

\begin{Prop}
    The discrete-spectrum displacement operators $U_{k, \alpha}$ and $V_{k, \beta}$ preserve the compressed space $\mathcal V_k$ and act on its basis $\{\ket{\lambda_{k,j}}_{q_k}\}$ as follows:
    \begin{align}
        U_{k, \alpha}\ket{\lambda_{k,j}}_{q_k} &= \exp(i\alpha\lambda_{j,k})\ket{\lambda_{k,j}}_{q_k},\label{eq:eqk}\\
       V_{k, \beta}\ket{\lambda_{k,j}}_{q_k}  &= \sum_{j''=0}^{k-1}\sum_{j'=0}^{k-1}I_k(\lambda_{k,j''}, \lambda_{k,j'})\overline{I_k(\lambda_{k,j}, \lambda_{k,j'})} \exp(-i\beta \lambda_{k,j'})\ket{\lambda_{k,j''}}_{q_k}.\label{eq:epk}
    \end{align}
\end{Prop}
\begin{proof}
    From the functional calculus for $\hat{q}_k$, $\hat{p}_k$, the operators $U_{k,\alpha}$, $V_{k,\beta}$ leave $\mathcal V_k$ invariant (and hence also reduce it since they are unitary). Further,  we have
    \begin{equation*}
        U_{k,\alpha} = \exp(i\alpha\hat{q}_k) = \sum_{j=0}^{k-1}\exp(i\alpha\lambda_{k,j})\ket{\lambda_{k,j}}_{q_k}\bra{\lambda_{k,j}}_{q_k}+P_k^\perp .
    \end{equation*}
    As the $k$-compressed subspace is reducing for operators $U_{k,\alpha}$ and $V_{k, \beta}$, we will determine their representation with respect to a basis consisting of the eigenvectors of the compressed position operators. Indeed, recall we have from the change of basis matrix noted above, 
\begin{equation*}
    \ket{\lambda_{k,j}}_{q_k} = \sum_{j'=0}^{k-1}{}_{p_k}\bk{\lambda_{k,j'}}{\lambda_{k,j}}_{q_k} \ket{\lambda_{k,j'}}_{p_k} . 
\end{equation*}
Then, for $\beta\in \mathbb R$, we have  
\begin{align*}
  V_{k,\beta} \ket{\lambda_{k,j}}_{q_k} = \exp(-i\beta \hat{p}_k)\ket{\lambda_{k,j}}_{q_k} & = \sum_{j'=0}^{k-1}{}_{p_k}\bk{\lambda_{k,j'}}{\lambda_{k,j}}_{q_k} \exp(i\beta \lambda_{k,j'})\ket{\lambda_{k,j'}}_{p_k}\\
    &= \sum_{j''=0}^{k-1}\sum_{j'=0}^{k-1}{}_{q_k}\bk{\lambda_{k,j''}}{\lambda_{k,j'}}_{p_k}{}_{p_k}\bk{\lambda_{k,j'}}{\lambda_{k,j}}_{q_k} \exp(-i\beta \lambda_{k,j'})\ket{\lambda_{k,j''}}_{q_k}. 
\end{align*}
Again, using the change of basis matrix $\{{}_{q_k}\bk{\lambda_{k,j}}{\lambda_{k,i}}_{p_k}\}_{i,j}$ expressed in terms of the Hermite polynomials through the partial sum of the Melher Kernel function as in \Cref{eq:changeofbasis}, completes the proof.
\end{proof}

Note that the operators $U_{k, \alpha}$ and $V_{k, \beta}$ do not commute, regardless of the value of $\alpha$ and $\beta$, which can be seen to follow from Proposition~\ref{compressedcommutator}. However, the discrete-spectrum displacement operators converge to the actual displacement operators that may commute, which is an important feature used in quantum error correction \cite{gkp}.

\begin{Prop}
    For any fixed  $\alpha$, $\beta$ the sequences of $U_{k, \alpha}$ and  $V_{k, \beta}$ converge respectively to the operators $\exp(i\alpha \hat{q})$ and $\exp(-i\beta \hat{p})$ in the norm topology. 
\end{Prop}
\begin{proof}
By Trotter's theorem \cite[Theorem VIII.21]{reedsimonmethods}, the operators $e^{i\alpha \hat{q}_k}$ converge in the strong operator topology to the operator $e^{i\alpha \hat{q}}$. In other words, for all $\ket{\psi}\in l^2$, we have 
$
    e^{i\alpha \hat q}\ket{\psi} = \lim_{k\to \infty} e^{i\alpha \hat q_k}\ket{\psi}.
$
In addition, the operators $e^{i\alpha \hat{q}_k}$ are all bounded, hence applying the uniform boundedness theorem, the sequence of operators $e^{i\alpha \hat{q}_k}$ converges in norm to $e^{i\alpha\hat{q}}$. This argument works as well in proving that $V_{k, \beta}$ converges to the operator $\exp(-i\beta \hat{p})$ in the norm topology.
\end{proof}

A consequence of the norm convergence above is that the commutator $[U_{k,\alpha}, V_{k, \beta}]$ of the operators $U_{k,\alpha}$ and $V_{k, \beta}$ converges to $[\exp(i\alpha \hat{q}), \exp(-i\beta \hat{p})]$. Now, we consider the case where the points of convergence commute, which occurs for $\alpha$ and $\beta$ verifying the equation $\alpha = \beta = 2\sqrt{\pi}$. In \cite{gkp}, such operators generate an abelian group that stabilizes a $2$-dimensional (functional) space with basis 
\begin{equation*} 
\bra{0} = \sum_{m = -\infty}^{\infty}\bra{\, 2m\sqrt{\pi}\, }_q \quad \mathrm{and} \quad 
\bra{1} = \sum_{m = -\infty}^{\infty}\bra{\, (2m+1)\sqrt{\pi} \,}_q . 
\end{equation*} 
In other words, these states are generalized $+1$-eigenstates (in the sense we discussed above) of the  displacement operators $\exp(i2\sqrt{\pi} \hat{q})$ and $\exp(-i2\sqrt{\pi} \hat{p})$. Returning to our converging operators $U_{2\sqrt{\pi}, k}$  and $V_{2\sqrt{\pi}, k}$, we would like to show that they approximately (as they do not commute) stabilize a two-dimensional subspace. 

Note that the operators $U_{2\sqrt{\pi}, k}-I$  and $V_{2\sqrt{\pi}, k}-I$ converge, as $k\rightarrow \infty$, respectively to $\exp(i2\sqrt{\pi} \hat{q})-I$ and $\exp(-i2\sqrt{\pi} \hat{p})-I$ in the norm topology. As this implies convergence in the weak operator topology, for all $\ket{x}\in s \subset \ell_2$ and a fixed basis state $\ket{n}$ we have  
\begin{equation}
\lim_{k \rightarrow \infty}    \bra{n}(U_{2\sqrt{\pi}, k}-I)\ket{x} = \bra{n}(\exp(-i2\sqrt{\pi} \hat{q})-I)\ket{x} \mbox{ and } \lim_{k \rightarrow \infty} \bra{n}(V_{2\sqrt{\pi}, k}-I)\ket{x}=  \bra{n}(\exp(-i2\sqrt{\pi} \hat{p})-I)\ket{x}. \label{eq:weak_convergence}
\end{equation}
Now, for $k\geq 1$, we define the states $\ket{0_k}$ and $\ket{1_k}$ in $\mathcal V_k$, which are truncated versions of the (generalized) vectors $\ket{0}$ and $\ket{1}$. These are given by   
\begin{align*}
    \ket{0_k} & =\sum_{n=0}^{k-1} \sum_{m=-\infty}^{\infty}\bk{n}{ 2m\sqrt{\pi}}_q\ket{n} =  \sum_{n=0}^{k-1} \sum_{m=-\infty}^{\infty}\phi_n( 2m\sqrt{\pi})\ket{n},\\
    \ket{1_k} & =\sum_{n=0}^{k-1} \sum_{m=-\infty}^{\infty}\bk{n}{(2m+1)\sqrt{\pi}}_q\ket{n}=  \sum_{n=0}^{k-1} \sum_{m=-\infty}^{\infty}\phi_n( (2m+1)\sqrt{\pi})\ket{n},
\end{align*}
where we have also recalled the connection with the Hermite polynomials in these expressions. Using \Cref{eq:weak_convergence} and the norm convergence of the unitary operators, we have the following asymptotic relations:   
\begin{equation*}
     \bra{0_k}(U_{2\sqrt{\pi}, k}-I)\ket{x} =   \bra{0_k}(\exp(i2\sqrt{\pi} \hat{q})-I)\ket{x} \mbox{ and } \bra{1_k}(V_{2\sqrt{\pi}, k}-I)\ket{x} =   \bra{1_k}(\exp(-i2\sqrt{\pi} \hat{p})-I)\ket{x}.
\end{equation*}
Taken together, this gives us 
\begin{equation}
     \lim_{k\to \infty}\bra{0_k}(U_{2\sqrt{\pi}, k}-I)\ket{x}\to  \bra{0}(\exp(i2\sqrt{\pi} \hat{q})-I)\ket{x} \mbox{ and } \lim_{k\to \infty}\bra{1_k}(V_{2\sqrt{\pi}, k}-I)\ket{x}\to  \bra{1}(\exp(-i2\sqrt{\pi} \hat{p})-I)\ket{x}.\label{eq:convergence}
\end{equation}
Indeed, to start, note that the right-hand side of the above convergence expressions is well defined. In addition, the operator $\exp(i2\sqrt{\pi} \hat{q})$ (respectively $\exp(i2\sqrt{\pi} \hat{p})$) maps the Schwartz space back to itself as it multiplies (respectively shifts) the Schwartz function associated to an element in $s$. Further, the sequence of linear functionals $\bra{0_k} = \bra{0} P_k$ (respectively $\bra{1_k}$) converges strongly (i.e., pointwise) on $s$ to the functional $\bra{0}$ (respectively $\bra{1}$).

We can thus state the following result. 

\begin{Prop}
    The truncated states $\ket{0_k}$ and $\ket{1_k}$, for $k\geq 1$, define a $2$-dimensional subspace of the compressed space $\mathcal V_k$ that is approximately stabilized by the operators $U_{2\sqrt{\pi}, k}$  and $V_{2\sqrt{\pi}, k}$ in the sense given by the limits of  \Cref{eq:convergence}.
\end{Prop}

\begin{Rm}
 A central class of error correcting codes in photonic quantum computing, the GKP codes \cite{gkp}, are built as infinite-dimensional extensions of finite-dimensional codes. In these codes, as noted above, a specific set of (commuting) displacement operators $\exp(i\alpha \hat{q}), \exp(-i\beta \hat{p})$ act as stabilizers, while others serve as logical (generalized) Pauli operators. Additionally, displacement operators, also called Weyl operators, form an orthogonal complete set on the space of operators on $L^2(\mathbb R)$, and so can be used to describe the noise processes affecting the system. The goal to view GKP codes as finite-dimensional limits, appropriately formulated, gave the initial impetus for this work and is an ongoing investigation. We expect the calculations presented in this section will be useful in that context. 

 We further note that techniques from matrix and operator theory developed to study `higher rank numerical ranges' \cite{choi2006higher,choi2006quantum,gau2011higher,martinez2008higher,woerdeman2008higher}, which were originally motivated by (finite-dimensional) quantum error correction, could be useful in this setting. Such numerical ranges are obtained through compressions of certain classes of operators, as is the case in this work. 
\end{Rm}

\vspace{0.1in} 

{\noindent}{\bf Acknowledgements.} We are grateful to Rafael N. Alexander for helpful conversations. S.A. was supported by a Mitacs Accelerate Fellowship. D.W.K. was supported by NSERC Discovery Grant RGPIN-2024-400160. R.P. was supported by the NSERC Discovery Grant RGPIN-2022-04149.

\bibliographystyle{plain}
\bibliography{refs}

@article{trzetrzelewski2003quantum,
  title={Quantum systems in a cut {F}ock space},
  author={Trzetrzelewski, Maciej and Wosiek, Jacek},
  journal={Acta Physica Polonica B},
  volume = {35},
  year={2004},
  pages = {1615--1624},
}

@article{deNeeve,
	author = {de Neeve, Brennan and Nguyen, Thanh-Long and Behrle, Tanja and Home, Jonathan P.},
	date = {2022/03/01},
	doi = {10.1038/s41567-021-01487-7},
	id = {de Neeve2022},
	isbn = {1745-2481},
	journal = {Nature Physics},
	number = {3},
	pages = {296--300},
	title = {Error correction of a logical grid state qubit by dissipative pumping},
	url = {https://doi.org/10.1038/s41567-021-01487-7},
	volume = {18},
	year = {2022}}

@article{fluhmann,
	author = {Fl{\"u}hmann, C. and Nguyen, T. L. and Marinelli, M. and Negnevitsky, V. and Mehta, K. and Home, J. P.},
	date = {2019/02/01},
	doi = {10.1038/s41586-019-0960-6},
	id = {Fl{\"u}hmann2019},
	isbn = {1476-4687},
	journal = {Nature},
	number = {7745},
	pages = {513--517},
	title = {Encoding a qubit in a trapped-ion mechanical oscillator},
	url = {https://doi.org/10.1038/s41586-019-0960-6},
	volume = {566},
	year = {2019}}

@article{doi:10.1126/science.adk7560,
	author = {Shunya Konno and Warit Asavanant and Fumiya Hanamura and Hironari Nagayoshi and Kosuke Fukui and Atsushi Sakaguchi and Ryuhoh Ide and Fumihiro China and Masahiro Yabuno and Shigehito Miki and Hirotaka Terai and Kan Takase and Mamoru Endo and Petr Marek and Radim Filip and Peter van Loock and Akira Furusawa},
	doi = {10.1126/science.adk7560},
	eprint = {https://www.science.org/doi/pdf/10.1126/science.adk7560},
	journal = {Science},
	number = {6680},
	pages = {289-293},
	title = {Logical states for fault-tolerant quantum computation with propagating light},
	url = {https://www.science.org/doi/abs/10.1126/science.adk7560},
	volume = {383},
	year = {2024}}

@article{PhysRevLett.132.150607,
  title = {Autonomous Quantum Error Correction of {G}ottesman-{K}itaev-{P}reskill States},
  author = {Lachance-Quirion, Dany and Lemonde, Marc-Antoine and Simoneau, Jean Olivier and St-Jean, Lucas and Lemieux, Pascal and Turcotte, Sara and Wright, Wyatt and Lacroix, Am\'elie and Fr\'echette-Viens, Jo\"elle and Shillito, Ross and Hopfmueller, Florian and Tremblay, Maxime and Frattini, Nicholas E. and Camirand Lemyre, Julien and St-Jean, Philippe},
  journal = {Physical Review Letters},
  volume = {132},
  issue = {15},
  pages = {150607},
  numpages = {6},
  year = {2024},
  month = {Apr},
  publisher = {American Physical Society},
  doi = {10.1103/PhysRevLett.132.150607},
  url = {https://link.aps.org/doi/10.1103/PhysRevLett.132.150607}
}

@article{PRXQuantum.3.030301,
  title = {Robust Preparation of Wigner-Negative States with Optimized SNAP-Displacement Sequences},
  author = {Kudra, Marina and Kervinen, Mikael and Strandberg, Ingrid and Ahmed, Shahnawaz and Scigliuzzo, Marco and Osman, Amr and Lozano, Daniel P\'erez and Thol\'en, Mats O. and Borgani, Riccardo and Haviland, David B. and Ferrini, Giulia and Bylander, Jonas and Kockum, Anton Frisk and Quijandr\'{\i}a, Fernando and Delsing, Per and Gasparinetti, Simone},
  journal = {PRX Quantum},
  volume = {3},
  issue = {3},
  pages = {030301},
  numpages = {12},
  year = {2022},
  month = {Jul},
  publisher = {American Physical Society},
  doi = {10.1103/PRXQuantum.3.030301},
  url = {https://link.aps.org/doi/10.1103/PRXQuantum.3.030301}
}

@article{sivak,
	author = {Sivak, V. V. and Eickbusch, A. and Royer, B. and Singh, S. and Tsioutsios, I. and Ganjam, S. and Miano, A. and Brock, B. L. and Ding, A. Z. and Frunzio, L. and Girvin, S. M. and Schoelkopf, R. J. and Devoret, M. H.},
	date = {2023/04/01},
	doi = {10.1038/s41586-023-05782-6},
	id = {Sivak2023},
	isbn = {1476-4687},
	journal = {Nature},
	number = {7955},
	pages = {50--55},
	title = {Real-time quantum error correction beyond break-even},
	url = {https://doi.org/10.1038/s41586-023-05782-6},
	volume = {616},
	year = {2023}}

@article{Eickbusch,
	author = {Eickbusch, Alec and Sivak, Volodymyr and Ding, Andy Z. and Elder, Salvatore S. and Jha, Shantanu R. and Venkatraman, Jayameenakshi and Royer, Baptiste and Girvin, S. M. and Schoelkopf, Robert J. and Devoret, Michel H.},
	date = {2022/12/01},
	doi = {10.1038/s41567-022-01776-9},
	id = {Eickbusch2022},
	isbn = {1745-2481},
	journal = {Nature Physics},
	number = {12},
	pages = {1464--1469},
	title = {Fast universal control of an oscillator with weak dispersive coupling to a qubit},
	url = {https://doi.org/10.1038/s41567-022-01776-9},
	volume = {18},
	year = {2022}}

@article{ibarcq,
	author = {Campagne-Ibarcq, P. and Eickbusch, A. and Touzard, S. and Zalys-Geller, E. and Frattini, N. E. and Sivak, V. V. and Reinhold, P. and Puri, S. and Shankar, S. and Schoelkopf, R. J. and Frunzio, L. and Mirrahimi, M. and Devoret, M. H.},
	date = {2020/08/01},
	doi = {10.1038/s41586-020-2603-3},
	id = {Campagne-Ibarcq2020},
	isbn = {1476-4687},
	journal = {Nature},
	number = {7821},
	pages = {368--372},
	title = {Quantum error correction of a qubit encoded in grid states of an oscillator},
	url = {https://doi.org/10.1038/s41586-020-2603-3},
	volume = {584},
	year = {2020}}

@article{Bourassa2021blueprintscalable,
  doi = {10.22331/q-2021-02-04-392},
  url = {https://doi.org/10.22331/q-2021-02-04-392},
  title = {Blueprint for a {S}calable {P}hotonic {F}ault-{T}olerant {Q}uantum {C}omputer},
  author = {Bourassa, J. Eli and Alexander, Rafael N. and Vasmer, Michael and Patil, Ashlesha and Tzitrin, Ilan and Matsuura, Takaya and Su, Daiqin and Baragiola, Ben Q. and Guha, Saikat and Dauphinais, Guillaume and Sabapathy, Krishna K. and Menicucci, Nicolas C. and Dhand, Ish},
  journal = {{Quantum}},
  issn = {2521-327X},
  publisher = {{Verein zur F{\"{o}}rderung des Open Access Publizierens in den Quantenwissenschaften}},
  volume = {5},
  pages = {392},
  month = feb,
  year = {2021}
}

@article{BRADY2024100496,
	author = {Anthony J. Brady and Alec Eickbusch and Shraddha Singh and Jing Wu and Quntao Zhuang},
	doi = {https://doi.org/10.1016/j.pquantelec.2023.100496},
	issn = {0079-6727},
	journal = {Progress in Quantum Electronics},
	pages = {100496},
	title = {Advances in bosonic quantum error correction with {G}ottesman--{K}itaev--{P}reskill Codes: Theory, engineering and applications},
	url = {https://www.sciencedirect.com/science/article/pii/S0079672723000459},
	volume = {93},
	year = {2024}}

@article{log-gates-GKP,
	author = {Matsos, V. G. and Valahu, C. H. and Millican, M. J. and Navickas, T. and Kolesnikow, X. C. and Biercuk, M. J. and Tan, T. R.},
	date = {2025/10/01},
	doi = {10.1038/s41567-025-03002-8},
	id = {Matsos2025},
	isbn = {1745-2481},
	journal = {Nature Physics},
	number = {10},
	pages = {1664--1669},
	title = {Universal quantum gate set for {G}ottesman--{K}itaev--{P}reskill logical qubits},
	url = {https://doi.org/10.1038/s41567-025-03002-8},
	volume = {21},
	year = {2025}}

@article{IC-gkp,
	author = {Larsen, M. V. and Bourassa, J. E. and Kocsis, S. and Tasker, J. F. and Chadwick, R. S. and Gonz{\'a}lez-Arciniegas, C. and Hastrup, J. and Lopetegui-Gonz{\'a}lez, C. E. and Miatto, F. M. and Motamedi, A. and Noro, R. and Roeland, G. and Baby, R. and Chen, H. and Contu, P. and Di Luch, I. and Drago, C. and Giesbrecht, M. and Grainge, T. and Krasnokutska, I. and Menotti, M. and Morrison, B. and Puviraj, C. and Rezaei Shad, K. and Hussain, B. and McMahon, J. and Ortmann, J. E. and Collins, M. J. and Ma, C. and Phillips, D. S. and Seymour, M. and Tang, Q. Y. and Yang, B. and Vernon, Z. and Alexander, R. N. and Mahler, D. H.},
	date = {2025/06/01},
	doi = {10.1038/s41586-025-09044-5},
	id = {Larsen2025},
	isbn = {1476-4687},
	journal = {Nature},
	number = {8068},
	pages = {587--591},
	title = {Integrated photonic source of {G}ottesman--{K}itaev--{P}reskill qubits},
	url = {https://doi.org/10.1038/s41586-025-09044-5},
	volume = {642},
	year = {2025}}

@article{gkp,
  title = {Encoding a qubit in an oscillator},
  author = {Gottesman, Daniel and Kitaev, Alexei and Preskill, John},
  journal = {Physical Review A},
  volume = {64},
  issue = {1},
  pages = {012310},
  numpages = {21},
  year = {2001},
  month = {Jun},
  publisher = {American Physical Society},
  doi = {10.1103/PhysRevA.64.012310},
  url = {https://link.aps.org/doi/10.1103/PhysRevA.64.012310}
}

@incollection{dominici2010asymptotic,
  title={Asymptotic Analysis of the Zeros of {H}ermite Polynomials},
  author={Dominici, Diego},
  booktitle={Progress in Industrial Mathematics at ECMI 2008},
  pages={99--103},
  year={2010},
  publisher={Springer}
}

@book{szeg1939orthogonal,
  title={Orthogonal polynomials},
  author={Szeg\H{o}, G\'{a}bor},
  volume={23},
  year={1939},
  publisher={American Mathematical Soc.}
}

@book{brezis2011functional,
  title={Functional analysis, Sobolev spaces and partial differential equations},
  author={Br{\'e}zis, Haim},
  volume={2},
  year={2011},
  publisher={Springer}
}

@article{laplace1810memoire,
  title={M{\'e}moire sur les int{\'e}grales d{\'e}finies et leur application aux probabilit{\'e}s, et sp{\'e}cialement a la recherche du milieu qu’il faut choisir entre les r{\'e}sultats des observations},
  author={Laplace, Pierre-Simon},
  journal={Mem. Acad. Sci.(I), XI, Section V},
  pages={375--387},
  year={1810}
}

@book{marquis1820theorie,
  title={Th{\'e}orie analytique des probabilit{\'e}s},
  author={Laplace, Pierre Simon},
  volume={7},
  year={1820},
  publisher={Courcier}
}

@article{chebyshev1859developpement,
  title={Sur le d{\'e}veloppement des fonctions {\`a} une seule variable},
  author={Chebyshev, Pafnutii Lvovich},
  journal={Bull. Acad. Sci. St. Petersb},
  volume={1},
  number={193-200},
  pages={124},
  year={1859}
}

@book{hermite1864nouveau,
  title={Sur un nouveau d{\'e}veloppement en s{\'e}rie des fonctions},
  author={Hermite, M},
  year={1864},
  publisher={Imprimerie de Gauthier-Villars}
}

@article{moretti2017spectral,
  title={Spectral theory and quantum mechanics},
  author={Moretti, Valter},
  journal={UNITEXT, Italy: Springer International Publishing AG},
  year={2017},
  publisher={Springer}
}

@book{reedsimonmethods,
  title={Methods of modern mathematical physics: Functional analysis},
  author={Reed, Michael and Simon, Barry},
  year={1972},
  publisher={Elsevier}
}

@book{chihara2011introduction,
  title={An introduction to orthogonal polynomials},
  author={Chihara, Theodore S.},
  year={2011},
  publisher={Courier Corporation}
}

@book{borwein2012polynomials,
  title={Polynomials and polynomial inequalities},
  author={Borwein, Peter and Erd{\'e}lyi, Tam{\'a}s},
  volume={161},
  year={2012},
  publisher={Springer Science \& Business Media}
}

@article{choi2006higher,
  title={Higher-rank numerical ranges and compression problems},
  author={Choi, Man-Duen and Kribs, David W and {\.Z}yczkowski, Karol},
  journal={Linear algebra and its applications},
  volume={418},
  number={2-3},
  pages={828--839},
  year={2006},
  publisher={Elsevier}
}

@article{choi2006quantum,
  title={Quantum error correcting codes from the compression formalism},
  author={Choi, Man-Duen and Kribs, David W and {\.Z}yczkowski, Karol},
  journal={Reports on Mathematical Physics},
  volume={58},
  number={1},
  pages={77--91},
  year={2006},
  publisher={Elsevier}
}

@article{martinez2008higher,
  title={Higher-rank numerical range in infinite-dimensional {H}ilbert space},
  author={Martinez-Avendano, Ruben},
  journal={Operators and Matrices},
    volume={2},
  number={},
  pages={249--264},
  year={2008}
}

@article{gau2011higher,
  title={Higher rank numerical ranges of normal matrices},
  author={Gau, Hwa-Long and Li, Chi-Kwong and Poon, Yiu-Tung and Sze, Nung-Sing},
  journal={SIAM Journal on Matrix Analysis and Applications},
  volume={32},
  number={1},
  pages={23--43},
  year={2011},
  publisher={SIAM}
}

@article{woerdeman2008higher,
  title={The higher rank numerical range is convex},
  author={Woerdeman, Hugo J},
  journal={Linear and Multilinear Algebra},
  volume={56},
  number={1-2},
  pages={65--67},
  year={2008},
  publisher={Taylor \& Francis}
}

\end{document}